\documentclass[11pt]{article}

\usepackage{fullpage,times}
\usepackage{amsmath,amsfonts,amsthm}
\usepackage{multicol}
\usepackage{url}
\usepackage{graphicx,xspace}

\newcommand{\citep}[1]{\cite{#1}}

\begin{document}

\title{Online Stochastic Packing applied to  \\ Display Ad Allocation}

\author{
Jon Feldman\thanks{Google Research, 76 9th Ave,
New York, NY 10011, \{jonfeld,mirrokni,cstein\}@google.com }
\and Monika Henzinger \thanks{University of Vienna, Austria.
  monika.henzinger@univie.ac.at}
\and Nitish Korula \thanks{University of Illinois at Urbana Champaign,
  nkorula2@illinois.edu} 
\and Vahab S. Mirrokni\footnotemark[1] 
\and Cliff Stein\thanks{Google Research and Columbia University, New
  York, NY} 
}

\newtheorem{fact}{Fact}
\newtheorem{lemma}{Lemma}
\newtheorem{theorem}[lemma]{Theorem}
\newtheorem{assumption}[lemma]{Assumption}
\newtheorem{definition}[lemma]{Definition}
\newtheorem{corollary}[lemma]{Corollary}
\newtheorem{prop}[lemma]{Proposition}
\newtheorem{claim}[lemma]{Claim}
\newtheorem{remark}[lemma]{Remark}
\newtheorem{prob}{Problem}
\newtheorem{conjecture}{Conjecture}
\newtheorem{proposition}[lemma]{Proposition}

\newenvironment{proofsketch}{\noindent{\sc Proof Sketch.}}%
        {\hspace*{\fill}$\Box$\par\vspace{4mm}}
\newenvironment{proofof}[1]{\smallskip\noindent{\bf Proof of #1.}}%
        {\hspace*{\fill}$\Box$\par}

\newcommand{\opt}{\textsc{OPT}}
\newcommand{\etal}{{\em et al.}\ }
\newcommand{\dual}{\textbf{Dual-LP}\xspace}
\newcommand{\primal}{\textbf{Primal-LP}\xspace}
\def\Exp#1{\mathbb{E}\left[#1\right]}
\def\Prob#1{\textbf{Pr}\left[#1\right]}

\def\eps{\varepsilon}
\def\bar{\overline}
\def\floor#1{\lfloor {#1} \rfloor}
\def\ceil#1{\lceil {#1} \rceil}
\def\script#1{\mathcal{#1}}
\def\prof#1{\textrm{profit}(#1)}
\def\stdLP{\textbf{Standard LP}}
\def\P{\script{P}}
\def\dmax{d_{\max}}
\def\dmin{d_{\min}}
\def\pdavg{$\mathsf{PD\_AVG}$\xspace}
\def\greedy{$\mathsf{GREEDY}$\xspace}
\def\pdexp{$\mathsf{PD\_EXP}$\xspace}
\def\dualsample{$\mathsf{D\_SAMPLE}$\xspace}
\def\dualbase{$\mathsf{DualBase}$\xspace}
\def\dualhistory{$\mathsf{D\_HISTORY}$}
\def\dualbegin{$\mathsf{D\_BEGIN}$}
\def\hybrid{$\mathsf{HYBRID}$\xspace}
\def\fair{$\mathsf{FAIR}$\xspace}
\def\lpweight{$\mathsf{LP\_WEIGHT}$\xspace}
\def\optweight{\mathsf{OPT\_WEIGHT}}
\def\optcardinality{\mathsf{OPT\_CARDINALITY}}
\def\gain{\textit{gain}}
\def\O{\script{O}}

\maketitle
\begin{abstract}
  Inspired by online ad allocation, we study online stochastic packing
  linear programs from theoretical and practical standpoints.  We
  first present a near-optimal online algorithm for a general class of
  packing linear programs which model various online resource
  allocation problems including online variants of routing, ad
  allocations, generalized assignment, and combinatorial auctions.  As
  our main theoretical result, we prove that a simple primal-dual
  training-based algorithm achieves a $(1-o(1))$-approximation
  guarantee in the random order stochastic model. This is a
  significant improvement over logarithmic or constant-factor
  approximations for the adversarial variants of the same problems
  (e.g. factor $1-{1\over e}$ for online ad allocation, and $\log(m)$
  for online routing).
  We then focus on the online display ad allocation problem and study
  the efficiency and fairness of various training-based and online
  allocation algorithms
  on data sets collected from real-life display ad allocation system.
  Our experimental evaluation confirms the effectiveness of
  training-based primal-dual algorithms on real data sets, and
  also indicate an intrinsic trade-off between
  fairness and efficiency.
\end{abstract}

\section{Introduction}
Online stochastic optimization is a central problem in operations
research with many applications in dynamic resource allocation. In
these settings, given a set of resources, demands for the resources
arrive online, with associated values; given a general prior about the
demands, one has to decide whether and how to satisfy (i.e., allocate
the desired resources to) a demand when it arrives. The goal is to
find a valid assignment with maximum total value.  Such problems
appear in many areas including online
routing~\cite{BN06,AAP-Routing93}, online combinatorial
auctions~\cite{CHMS10}, online ad allocation
problems~\cite{msvv,devanur-hayes,WINE09}, and online dynamic pricing
and inventory management problems.  For example, in routing problems,
we are given a network with capacity constraints over edges; customers
arrive online and bid for a subset of edges (typically a path) in the
network, and the goal is to assign paths to new customers so as to
maximize the total social welfare. Similarly, in online combinatorial
auctions, bidders arrive online and may bid on a subset of resources;
the auctioneer should decide whether to sell those resources to the
bidder. In the display ads problem, when users visit a website, the
website publisher has to choose ads to show them so as to maximize the
value of the displayed ads.  In this paper, we study these online
stochastic resource allocation problems from theoretical and practical
standpoints.  Our theoretical results apply to a general set of
problems including all those discussed above.  Our practical results
apply to the problem of display ads and give additional validation of
our theoretical models and results.

More specifically, we consider the following general class of packing
linear programs (PLP): Let $J$ be a set of $m$ {\em resources}; each
resource $j \in J$ has a capacity $c_j$. The set of resources and
their capacities are known in advance. Let $I$ be a set of $n$ {\em
  agents} that arrive one by one online, each with a set of {\em
  options} $O_i$.  Each option $o \in O_i$ of agent $i$ has an
associated {\em value} $w_{io} \geq 0$ and requires $a_{ioj} \geq 0$
units of each resource $j$. The set of options and the values $w_{io}$
and $a_{ioj}$ arrive together with agent $i$.  When an agent arrives,
the algorithm has to immediately decide whether to assign the agent
and if so, which option to choose. The goal is to find a maximum-value
allocation that does not allocate more of any resource than is
available.

In the {\em adversarial} or worst-case setting, no online algorithm
can achieve any non-trivial competitive ratio; consider the simple
case of one resource with capacity one and two agents.  For each agent
there are just two options, namely to get the resource or not to get
it. If an agent gets the resource, he uses its whole capacity.  The
first agent has value 100 for getting the resource and value 0 for not
getting the resource. If he is assigned the resource, then the value
of the second agent for getting the resource is 10000, otherwise it is
1. In both cases the algorithm achieves less than 1/100th of the value
of the optimal solution.  This example can easily be generalized to
show that no non-trivial competetive ratio is possible.

Since in the adversarial setting the lack of prior information about the
arrival rate of different types of agents implies strong impossibility
results, it is natural to consider {\em stochastic}
settings for online allocation problems, where we may have some prior
information about the arrival rate of different types of agents. In
particular, we consider the {\em random-order stochastic model}, where
the order in which impressions arrive is random, but we do not have
any other prior information.  We present a training-based online
algorithm for the general class of packing linear programs described
above and prove that in the random-order stochastic model, it achieves
an approximation ratio of $1-\eps$ under some mild
assumptions\footnote{In this context, an
  ``$\alpha$-approximation'' means that with high probability under
  the randomness in the stochastic model, the algorithm achieves at
  least an $\alpha$ fraction of the value (efficiency) of the offline optimal
  solution for the actual instance.}. 
  This result also implies the same result in the
i.i.d. model\footnote{In the i.i.d model each impression arrives
  independently and identically according to a particular but unknown
  probability distribution over the set of possible types of
  impressions~\cite{GM07}. Our stochastic model captures the i.i.d model.}.

Our {\em training-based primal-dual} algorithm for the stochastic PLP
problem observes the first $\eps$ fraction of the input and then
solves an LP on this instance. (This requires knowing the number of
agents in advance, which is unavoidable for any sub-logarithmic
approximation; see Theorem~\ref{thm:lowerBound}.) For
each resource, the corresponding dual variable extracted from this LP
serves as a \emph{(posted) price per unit} of the resource for the
remaining agents. The algorithm allocates to each remaining agent the
option maximizing his {\em utility}, defined as the difference between
the value of an option and the price he must pay to obtain the
necessary resources.  We prove that this algorithm provides a $1-\eps$
approximation for the large class of natural packing problems we
consider, provided that no individual option for any agent consumes
too much of any resource or provides too large a fraction of the total
value. Specifically we show the following result. Recall that $n$ and
$m$ denote the number of agents and resources respectively; $q$
denotes $\max_i |O_i|$ and $\opt$ the value of an optimal off-line
solution to the PLP problem.

\begin{theorem}\label{thm:main}
  The Training-Based Primal-Dual algorithm is $(1 -
  O(\eps))$-competitive (a PTAS) for the online stochastic PLP
  problem with high probability, as long as  (1) $\max_{i,o}\left
    \{\frac{w_{io}}{\opt}\right\} \le \frac{\eps}{(m+1) (\ln n + \ln
    q)}$ and (2)  $\max_{i,o,j}\left\{\frac{a_{ioj}}{c_j}\right\} \le
  \frac{\eps^3}{(m+1) (\ln n + \ln q)}$.
\end{theorem}

\subsection{Applications}
Theorem~\ref{thm:main} has many applications; we elaborate on several,
including routing problems, online combinatorial auctions, the display
ad problem, and the adword allocation problem.  For each of these
problems, we improve on the known results for the online version.  In
each, we will comment on the interpretation of the two conditions of
Theorem~\ref{thm:main} in that application.

\medskip
In the \emph{online routing} problem, we are initially given a network
with capacity constraints over the $m$ edges. When a customer $i \in
I$ arrives online, she wishes to send $d_i$ units of flow between some
vertices $s_i$ and $t_i$, and derives $w_i$ units of value from
sending such flow. Thus, the set of options $O_i$ for customer $i$ is
the set of all $s_i-t_i$ paths in the network.  The algorithm must
pick a set of customers $I^* \subseteq I$, and satisfy their demands
by allocating a path to each of them while respecting the capacity
constraints on each edge; the goal is to maximize the total value of
satisfied customers. For this problem, the dual variables learned from
the sample yield a price for each edge; each customer is allocated the
minimum-cost $s_i-t_i$ path if its cost is no more than $w_i$. In road
networks, for instance, these dual variables can be interpreted as the
tolls to be charged to prevent congestion.  Theorem~\ref{thm:main}
applies when the contributions of individual agents/vehicles to the
total objective or to road congestion are small.  As one such example,
over a million vehicles enter or leave Manhattan daily, with the
George Washington Bridge alone carrying several hundred thousand.
Online routing problems have been studied extensively in the
adversarial model when demands can be large, and there are
\emph{(poly)-logarithmic} lower and upper bounds even for special
cases \cite{AAP-Routing93,BN06}.  Our approach gives a $1 -
o(1)$-approximation for the described stochastic variants of these
problems.

\medskip
In the \emph{combinatorial auction} problem, we are initally given a
set $J$ of $m$ goods, with $c_j$ units for each good $j \in J$.
Agents arrive online, and the options for agent $i$ may include
different bundles of goods he values differently; option $o \in O_i$
provides $w_{io}$ units of value, and requires $a_{ioj}$ units of good
$j$.  We wish to find a valid allocation maximizing social
welfare. Here, the dual variables learned from the sample yield a
price per unit of each good; each agent picks the option that
maximizes his utility.  Here Theorem~\ref{thm:main} applies as long as
no individual agent controls a large fraction of the market, and as
long as the set of options for any single agent is at most exponential
in the number of resources. These conditions often hold, as in cases
when bidders are single-minded or the number of bundles they are
interested in is polynomial in $n$, or if their options correspond to
using different subsets of the resources.  We also observe that the
posted prices result in a take-it-or-leave it auction, and thus a
truthful online allocation mechanism.
Revenue maximization in online auctions using sequence item pricing
has been explored recently in the literature~\cite{BIKK08,CHMS10}.
\cite{CHMS10} achieves constant-factor approximations for these
problems in more general models than we consider.

\medskip 
In the {\em Display Ads Allocation (DA)} problem~\cite{WINE09}, there
is a set $J$ of $m$ advertisers who have paid a web publisher for
their ads to be shown to visitors to the website. The contract bought
by advertiser $j$ specifies an integer upper bound on the number
$n(j)$ of impressions that $j$ is willing to pay for. A set $I$ of
impressions arrives online, each impression $i$ with a value $w_{ij}
\geq 0$ for advertiser $j$. Each impression can be assigned to at most
one advertiser, i.e., there are $m$ options for each impression, and
each option $o$ has $a_{ioj} = 1$ for advertiser $j$. The goal is to
maximize the value of all the assigned impressions. The dual variables
learned from the sample yield a discount factor $\beta_j$ for each
advertiser $j$, and the algorithm is to assign an impression to
advertiser $j$ that maximizes $w_{ij} - \beta_j$. The contracts for
advertisers typically involve thousands of impressions, so the
contribution of any one impression/agent is small, and the hypotheses
of Theorem~\ref{thm:main} hold.  The adversarial online DA problem was
considered in \cite{WINE09}, which showed that the problem was
inapproximable without exploiting {\em free disposal}; using this
property, a simple greedy algorithm is $1\over 2$-competitive, which
is optimal. When the demand of each advertiser is large, a $(1-{1\over
  e})$-competitive algorithm exists (see \cite{WINE09} for details of
the model and results), and this is the best possible. For the {\em
  unweighted} (max-cardinality) version of this problem in the
i.i.d. model, a $0.67$-competitive algorithm has been recently
developed~\cite{FMMM09}; this improves the known $1-{1\over
  e}$-approximation algorithm for online stochastic
matching~\cite{KVV}.

\medskip
The {\em AdWords (AW)} problem~\cite{msvv,devanur-hayes} is related to
the DA problem; here we allocate impressions resulting from search
queries. Advertiser $j$ has a budget $b(j)$ instead of a bound $n(j)$
on the number of impressions. Assigning impression $i$ to advertiser
$j$ consumes $w_{ij}$ units of $j$'s budget instead of 1 of the $n(j)$
slots, as in the DA problem. Several approximation algorithms have
been designed for the {\em offline} AW
problem~\cite{chakrabarty2008aba,srinivasan2008baf,azar2008iaa}.  For
the online setting, if every weight is very small compared to the
corresponding budget, there exist $(1-{1 \over e})$-competitive online
algorithms~\cite{msvv,buchbinder-jain-naor,GM08,AM09}, and this factor
is tight.  In order to go beyond the competitive ratio of $1-{1\over
  e}$ in the adversarial model, stochastic online settings have been
studied, such as the random order and i.i.d
models~\cite{GM08}. Devanur and Hayes~\cite{devanur-hayes} described a
primal-dual ($1-\eps$)-approximation algorithm for this problem in the
random order model, with the assumption that $\opt$ is larger than
$O({m^2 \over \eps^3})$ times each $w_{ij}$, where $m$ is the number
of advertisers; Theorem~\ref{thm:main} can be viewed as generalizing
this result to a much larger class of problems.

\subsection{Experimental Validation} 
For the applications described above, stochastic models are reasonable
as the algorithm often has an idea of what agents to expect. For
example, in the Display Ad Allocation problem, agents correspond to
users visiting the website of a publisher who has sold contracts to
advertisers.  As the publisher most likely sees similar user traffic
patterns from day to day, he has an idea of the available ad inventory
based on historical data.  In Section~\ref{sec:experimental}, we
perform preliminary experiments on real instances of the DA problem,
using actual display ad data for a set of anonymous publishers. As
with any real application, there are additional features of the
problem, and in the one we considered, both {\em fairness} and
efficiency were important metrics.  Hence, we also evaluated our
algorithms for fairness (see Section~\ref{sec:fair} for a precise
definition); we compared the efficiency and {fairness} of our
training-based algorithm with those of algorithms from \cite{WINE09}
designed for the adversarial setting, as well as hybrid algorithms
combining the two approaches.  We propose a new approach for
evaluating the fairness of an allocation, based on finding an
``ideal'' fair allocation, and measuring the distance to that
allocation.  Our experimental results validate Theorem~\ref{thm:main}
for this application, as they show that on this real data set,
training indeed helps efficiency by 5-12\%, and that the online
algorithms from~\cite{WINE09} are significantly better than a simple
greedy approach.

\subsection{Other Related Work} 

Our proof technique is similar to that of \cite{devanur-hayes} for the
AW problem; it is based on their observation that dual variables
satisfy the complemtary slackness conditions of the first $\eps$
fraction of impressions and {\em approximately} satisfy these
conditions on the entire set.  However, one key difference is that in
the AW problem, the coefficients for variable $x_{ij}$ in the linear
program are the {\em same} in both the constraint and the objective
function. That is, the contribution an impression makes to an
advertisers value is identical to the amount of budget it consumes. In
contrast, in the general class of packing problems that we study,
these coefficients are unrelated, which complicates the proof.

The random-order model has been considered for several problems, often
called \emph{secretary} problems. The elements arriving online are
often the ground set of an appropriate matroid, and the goal is to
find a maximum weight independent set in the matroids; such problems
include finding a maximum value set of $k$ elements
\cite{Kleinberg05}, or finding a maximum spanning forest in a graph
when edges appear online. Other secretary problems include finding a
maximum weight set of items that fits in a Knapsack. (See
\cite{BIKK08} for a survey of these and other results.)
Constant-competitive algorithms are known for these problems; without
additional assumptions (such as those of Theorem 1), no algorithm can
achieve a competitive ratio better than $1/e$.  Specifically for the
DA problem, the results of \cite{KP09} imply that the random-order
model permits a $1/8$-competitive algorithm even without using the
free disposal property or the conditions of Theorem~\ref{thm:main}.

There have been recent results regarding ad allocation strategies in
display advertising in hybrid settings with both contract-based
advertisers and spot market advertisers ~\cite{GRVZ09,GMPV09}.  Our
results in this paper may be interpreted as a class of {\em
  representative bidding strategies} that can be used on behalf of
contract-based advertisers competing with the spot market
bidders~\cite{GRVZ09}.  There are many other interesting problems in
ad serving systems related to information retrieval and data mining
\cite{broder-tutorial,broder-IR-1,broder-broadmatch} as well as
various optimal caching strategies~\cite{flavio1,flavio2}; our focus
in this paper is on online allocation problems.]

\bigskip
It was recently brought to our attention that subsequent to the
submission of an earlier version of this paper (including our main
result), similar results (obtained independently) were posted in a
working paper\cite{AWY09}.

\section{A Training-based PTAS}\label{sec:ptas}

In this section, we present the primal-dual training-based algorithm
for the online stochastic packing problem, and prove Theorem 1: That
is, under mild (practically-motivated) assumptions, the algorithm
achieves an approximation factor of $1-\eps$.

Our algorithm examines the first $\eps n$ agents in order
before solving a Linear Program to compute the posted prices used for
the remaining agents. This requires advance knowledge of the number of
agents that will arrive; Theorem~\ref{thm:lowerBound} at the end of
this section shows that this is unavoidable. Recall that there is a
set $I$ of ``agents''; agent $i \in I$ has a set of mutually exclusive
options $O_i$, and we use an indicator variable $x_{io}$ to denote
whether agent $i$ selects alternative $o \in O_i$. Each option for an
agent may have a different ``size'' in each constraint; we use
$a_{ioj}$ to denote the size in constraint $j$ of option $o$ for agent
$i$.
We use $w_{io}$ to denote the value from selecting option $o$ for agent
$i$, and $c_j$ is the ``capacity'' of constraint $j$.  That is, our
goal is to maximize $w^T x$ while picking at most one option for each
agent, and subject to $A x \le c$. Subsequently, we normalize $A, c$
such that $c$ is the all-1's vector, and write the (normalized) primal
linear program below.  We also use the dual linear program, which
introduces a variable $\beta_j$ for each constraint $j$.

\setlength{\columnsep}{0.5in} \setlength{\columnseprule}{0.3pt}
\begin{multicols}{2}
  \begin{center} \textbf{Primal-LP}
  \end{center}
  \begin{align*} 
     & \hspace{-0.35in} \max \sum_{i} \sum_{o \in O_i}w_{io} x_{io} & \\ 
    \sum_{o \in O_i} x_{io} \ \le & \quad 1 & (\forall \ i) \\
    \sum_{i,o} a_{ioj} x_{io} \ \le & \quad 1 & (\forall \ j) \\
    x_{io} \ \ge & \quad 0 & (\forall \ i,o) \\
  \end{align*}

  \begin{center} \textbf{Dual-LP}
  \end{center}
  \vspace{-0.15in}
  \begin{align*} 
    \min \sum_j \beta_j \ + & \sum_i z_i & \\
    z_i + \sum_j \beta_j a_{ioj}\ \ge & \quad w_{io} & (\forall i,o) \\
    \beta_j, z_i \ \ge & \quad 0 & (\forall i, j)
  \end{align*}
\end{multicols}

Let $n$ be the total number of agents, $q = \max_i |O_i|$ be the maximum
number of options for any agent, and $m$ be the number of
constraints. We say that the \emph{gain} from option $o \in O_i$ is
$w_{io} - \sum_j \beta^*_j a_{ioj}$.
The Training-Based Primal-Dual Algorithm proceeds as follows:

\begin{enumerate}
  \item Let $S$ denote the first $\eps n$ agents in the sequence. For
    the purposes purposes of analysis, these agents are not selected.
    (Our implementations may assign these impressions according to
    some online algorithm.)

  \item Solve the \dual on the agents in $S$, with the objective
    function containing the term $\eps \beta_j$ instead of $\beta_j$
    for each $j \in [m]$.  (This is equivalent to reducing the
    capacity of a constraint from 1 to $\eps$; we refer to this as a
    \emph{reduced instance}.) Let $\beta^*_j$ denote the value of the
    dual variable for constraint $j$ in this optimal solution.

  \item For each subsequent agent $i$, if there is an option $o$ with
    non-negative gain, select the option\footnote{Assume for
      simplicity that there are no ties, and so there is a unique such
      option. This can be effectively achieved by adding a random
      perturbation to the weights; we omit details from this extended
      abstract.} $o$ of maximum gain, and set $z_i = \gain(o)$.
\end{enumerate}

We will refer to a variant of this algorithm in
Section~\ref{sec:experimental} as the {\dualbase} algorithm.  The
intuition behind this algorithm is simple; the dual variables
$\beta^*_j$ can be thought of as specifying a value/size ratio
necessary for an option to be selected. An optimal choice for each
$\beta_j$ gives an optimal solution to the packing problem; this fact
is proven implicitly in the next section, where we further show that
with high probability, the optimal choice $\beta^*_j$ on the sample
$S$ leads to a near-optimal solution on the entire instance.  In the
following, let $w_{\max}= \max_{i,o}\{w_{io}\}$, and let $a_{\max} =
\max_{i,o,j}\{a_{ioj}\}$.

\subsection{Proof of Theorem~\ref{thm:main}} We now prove
Theorem~\ref{thm:main}, showing that the above training-based
algorithm is a polynomial-time approximation scheme. ( Proofs of some
claims are in Appendix~\ref{app:proofs}.)  Let $I^* \subseteq I$
denote the set of agents $i$ with some option $o$ having non-negative
gain, and let $\O^*$ denote the set of pairs $\{(i,o) \mid i \in I^*,
o = \arg \max_{o \in O(i)} \gain(o)\}$. We abuse notation by writing
$i \in \O^*$ if there exists $o \in O(i)$ such that $(i,o) \in
\O^*$. We use $\O^*(S)$ to denote $\O^* \cap S$; note that $\O^* -
\O^*(S)$ represents the options selected by the algorithm (for the
purposes of analysis, we do not select any options for agents in $S$).

Given a vector $\beta^*$, we obtain a feasible solution to \dual by
selecting for each item in $I^*$, the option $o$ such that $(i,o) \in \O^*$
and setting $z_i = \gain(o)$.

\begin{definition} 
  Let $W = \sum_{(i,o) \in \O^*} w_{io}$ be the total weight of
  selected options, and let $W(S) = \sum_{(i,o) \in \O^*(S)} w_{io}$.
  Let $C(j) = \sum_{(i,o) \in \O^*} a_{ioj}$ and $C(j,S) = \sum_{(i,o) \in
    \O^*(S)} a_{ioj}$.
\end{definition}

For any fixed vector $\beta^*$, $\O^*$ and hence $W$ and each $C(j)$
are independent of the choice of the sample $S$; the expected value of
$W(S)$ is $\eps W$, and that of $C(j,S)$ is $\eps
C(j)$.\footnote{Though $\beta^*$ depends on $S$, many distinct samples
  $S$ may lead to the same vector $\beta^*$. Also, we take
  expectations over \emph{all} choices of $S$, not just those leading
  to the given $\beta^*$.}  The main idea of the proof is that if
$\beta^*$ satisfies the complentary slackness conditions on the first
$\eps n$ impressions (being an optimal solution), w.h.p. it
\emph{approximately} satisfies these conditions on the entire
set. Thus, we conclude that the values of $W(S)$ and $C(j,S)$ are
likely to be close to their expectations.

\medskip \noindent
The following lemma proved by \cite{devanur-hayes}, an application of
the Chernoff-Hoeffding bounds, is of use:

\begin{lemma}[\cite{devanur-hayes}]\label{lem:Hoeffding} 
  Let $Y = \{Y_1, \ldots, Y_n\}$ be a set of real numbers, and let $0
  < \eps < 1$. Let $S$ be a random subset of $Y$ of size $\eps n$ and
  let $Y_S = \sum_{i \in S} Y_i$. For any $0 < \delta < 1$:
  \[\Prob{ \left| Y_S - \Exp{Y_S} \right| \ge \frac{2}{3} \|Y\|_\infty
    \ln \left(\frac{2}{\delta}\right) + \|Y\|_2 \sqrt{2\eps \ln
      \left(\frac{2}{\delta}\right)}} \le \delta
  \]
\end{lemma}

\begin{definition}\label{def:bounds}
  For a sample $S$ and $j \in [m]$, let $r_j(S) = |C(j,S) - \eps
  C(j)|$, and let $t(S) = |W(S) - \eps W|$.  When the context is
  clear, we will abbreviate $r_j(S)$ by $r_j$ and $t(S)$ by $t$.
  \vspace{-2mm}
  \begin{enumerate}
  \item The sample $S$ is \emph{$r_j$-bad} if:\\
    $r_j \ge (m+1) (\ln n + \ln q) a_{\max} + \sqrt{C(j)} \cdot
    \left(2 \sqrt{\eps (m+1) (\ln n + \ln q) a_{\max}} \right)$.

  \item The sample $S$ is \emph{$t$-bad} if:\\
    $t \ge (m+1) (\ln n + \ln q) w_{\max} + \sqrt{W} \cdot \left( 2
      \sqrt{ \eps (m+1) (\ln n + \ln q) w_{\max}}\right)$.
  \end{enumerate}
\end{definition}

\begin{lemma}\label{lem:bounds}
  $\Prob{S \textrm{ is $r_j$-bad}} \le \frac{1}{m \cdot (nq)^{m+1}}$ for
  each $j$, and $ \Prob{S \textrm{ is $t$-bad}} \le
  \frac{1}{(nq)^{m+1}}$.
\end{lemma}
\begin{proof}
  To prove the first of these results, we simply apply
  Lemma~\ref{lem:Hoeffding} with $Y_i = a_{ioj}$ if $i \in \O^*$ and
  $0$ otherwise; we use $\|Y\|_2 \le \sqrt{\|Y\|_1 a_{\max}}$.  By
  setting $\delta = \frac{1}{m \cdot n^{m+1}}$, we obtain the desired
  result. (The coefficients are larger than necessary to keep the
  expression simple.)

  The proof of the second result is essentially identical, and hence
  omitted.
\end{proof}

We argue below that if $S$ is \emph{not} $t$-bad or $r_j$-bad for any
$j$, we obtain a good solution. We use the following simple
proposition:

\begin{proposition}\label{prop:cNear1}
  Let $j \in [m]$ be a constraint such that $C(j,S) = \eps$. If $S$ is
  not $r_j$-bad, we have $1 - 2\eps \le C(j) \le 1+3(\eps + \eps^2)$.
\end{proposition}
\begin{proof}
  To prove the former inequality, we use $C(j,S) - \eps C(j) \le (m+1)
  (\ln nq) a_{\max} + \sqrt{C(j)} \cdot \\
  \left(2 \sqrt{\eps (m+1) (\ln nq) a_{\max}} \right)$. As $a_{\max}
  \le \eps^3/\left((m+1) (\ln nq) \right)$, we have $\eps - \eps C(j)
  \le \eps^3 + \sqrt{C(j)} \cdot 2 \eps^2$; simple algebra now yields
  the desired result. The proof of the upper bound is similar, and so
  omitted.
\end{proof}

\begin{lemma} If the sample $S$ is not $t$-bad or $r_j$-bad for any
  constraint $j$, the value of the options selected by the algorithm
  is $(1 - O(\eps)) \opt$.
\end{lemma}
\begin{proof} 
  Let $D = \sum_j \beta^*_j + \sum_{i \in \O^*} z_i$ be the value of
  the feasible dual solution obtained by setting $z_i = \gain(o)$ for
  each $(i,o) \in \O^*$; by weak duality, $D$ is an upper bound on
  $\opt$. We show that $\sum_{(i,o) \in \O^* - \O^*(S)} w_{io} \ge (1
  - O(\eps)) D$, which completes the proof.

  First, we show that $W \ge (1-2\eps) D$. Let $J_1$ denote the set of
  constraints $j \in m$ such that $\beta^*_j > 0$, and $J_2 = [m] -
  J_1$ be the set of constraints such that $\beta^*_j = 0$. For each
  constraint $j \in J_1$, complementary slackness and
  Proposition~\ref{prop:cNear1} imply that if $S$ is not
  $r_j$-bad,$C(j) \ge 1 - 2\eps$.
  \begin{eqnarray*}
    W &=& \sum_{(i,o) \in \O^*} w_{io} = \sum_{(i,o) \in \O^*} \left( z_i + \sum_{j}
      a_{ioj} \beta^*_j \right) = \sum_{i \in I^*} z_i + \sum_j
    \left(\beta^*_j \sum_{(i,o) \in \O^*} a_{ioj} \right)\\
    & = & \sum_{i\in I^*} z_i + \sum_j \beta^*_j C(j)  \ge
    \sum_{i \in I^*} z_i + \sum_{j \in J_1} \beta^*_j (1 - 2\eps) \ge (1-2\eps)D
  \end{eqnarray*}

  where the penultimate inequality follows from the fact that for $j
  \in J_2$, $\beta^*_j = 0$, and for each $j \in J_1$, $C(j) \ge 1 -
  2\eps$.
  
  Now, the total value obtained by the algorithm is $W - W(S)$ (as the
  options for agents in $S$ were not selected); as $S$ is not $t$-bad,
  we have $W(S) \le \eps W + (m+1) (\ln nq) w_{\max} + 2 \sqrt{W}
  \sqrt{ \eps (m+1) (\ln nq) \cdot w_{\max} }$. But we have $(m+1) (\ln
  nq) w_{\max} \le \eps \opt$, and hence $W(S) \le \eps W + \eps \opt
  + 2 \sqrt{W} \sqrt{\eps^2 \opt} \le O(\eps W)$.  That is, the value
  obtained by the algorithm is at least $(1 - O(\eps)) W$, which is $(1
  - O(\eps)) \opt$.
\end{proof}

Note that the options selected by the algorithm, as described above,
may not be feasible even if $S$ is not $r_j$-bad;
Proposition~\ref{prop:cNear1} only implies that $C(j) \le
1+3(\eps+\eps^2)$. Thus, we might violate some constraints by a small
amount. This is easily fixed: simply decrease the capacities of all
constraints by a factor of $1+O(\eps)$. This reduces the value of the
optimal solution by no more than the same factor, as we can scale down
each $x_{io}$ by this factor to obtain a feasible solution with the
reduced capacities. Though our algorithm might violate the reduced
capacity of constraint $j$ by a factor of $1+O(\eps)$, we respect the
original capacity when $S$ is not $r_j$-bad. Thus, when $S$ is not
$t$-bad or $r_j$-bad for any $j$, we obtain a feasible solution with
value $(1-O(\eps))\opt$.

\bigskip
Finally, Lemma~\ref{lem:bounds} implies that for any fixed $\beta^*$,
the probability that a random sample $S$ of impressions is bad is less
than $\frac{2}{(nq)^{m+1}}$. The following lemma shows that there are
at most $(nq)^m$ distinct choices for $\beta^*$; as a result, the
sample is good for any $\beta^*$ with high probability. Therefore,
with high probability, our algorithm returns a feasible solution with
value at least $(1-O(\eps)) \opt$, proving Theorem~\ref{thm:main}.

\begin{lemma}\label{lem:fewSolutions}
  There are fewer than $(nq)^m$ distinct solutions $\beta^*$ that are
  returned by the algorithm after step $2$.
\end{lemma}
\begin{proof}
  Recall that an optimal (vertex) solution to the \dual on the reduced
  instance is defined purely by the $m$-dimensional vector $\beta^*$.
  The polytope defined by optimal solutions $\beta^*$ is defined by
  the constraints of the \dual, projected down to linear inequalities
  in $m$ dimensions. Since there are at most $q$ such constraints for
  each of the $n$ agents, there are at most $\binom{nq}{m}$ possible
  vertices of the polytope defined by optimal solutions $\beta^*$.
\end{proof}

\begin{theorem}\label{thm:lowerBound}
  Even in the full-information model, where $n$ agents drawn
  i.i.d. from a \emph{known} distribution arrive online, there is no
  $o(\log n/\log \log n)$-approximation for the online stochastic PLP
  unless the number $n$ of draws from the distribution is known in
  advance.
\end{theorem}
\begin{proof}  
  The intuition behind this proof is simple: The distribution may
  contain agents with very high value, but that arrive with low
  probability.  If there are many draws from the distribution, it is
  likely that such agents will arrive, and so some amount of resources
  should be reserved for them. On the other hand, an algorithm that
  reserves resources for low-probability events will waste a large
  fraction of its resources if there are only a few draws from the
  distribution.

  \medskip
  Fix $T \gg 1$; consider a problem with $3 T \log T$ units of a
  single resource, and every agent wishing precisely 1 unit of this
  resource. There are $T$ types of agents; agents of type $i \in
  \{0,\ldots T-1\}$ have value $T^{2i}$ for receiving a unit of
  resource. The probability of drawing an agent of type $i$ is
  $\approx \frac{1}{T^{2i}}$.  (Normalize these probabilities so they
  sum up to $1$; this changes the probabilities by a factor of
  $\approx 1 - (1 \over T^2)$, which we ignore for ease of
  exposition.)  Thus, the distribution of agents is known to the
  algorithm in advance.

  However, the algorithm \emph{does not know} how many agents will be
  drawn from this distribution. Suppose the number of draws is $6 T
  \log T \cdot T^{2j}$, for some $j \in \{0,\ldots, T-1\}$. It is easy
  to see that there will be very likely be more than $3 T \log T$
  agents of type $j$, and no agents of type $j+1$ or higher. Thus, the
  optimal solution has value $3 T \log T \cdot T^{2j}$; the hypotheses
  of Theorem 1 will hold, as no item contributes too much to the value
  of the optimal solution or uses too much of the shared resource.
  
  \medskip 

  Now consider any deterministic algorithm. If, for any $k \le j$, it
  has selected fewer than $3 \log T$ agents of type $k$ after $6 T
  \log T \cdot T^{2k}$ draws, it has a solution of value less than $3
  \log T \cdot T^{2k}$ (from agents of type $k$) plus $3T \log T \cdot
  T^{2k-2}$, which is $3 \log T \cdot T^{2k} (1 + o(1))$; this is
  roughly a factor of $T$ smaller than the optimal solution, which has
  value $3 T \log T \cdot T^{2k}$.  Thus, to maintain a $o(T)$
  competitive ratio, it must have selected at least $3 \log T$ agents
  of type $k$ after $6T \log T \cdot T^{2k}$ draws, as there may be no
  subsequent agents. However, this implies that after $6T \log T \cdot
  T^{2(T-2)}$ draws, the algorithm must have selected at least one
  agent from each of types $\{0,\ldots, (T-2)\}$. But there are $3
  \log T \cdot (T-1)$ such agents that must have been selected, each
  using a unit of the resource.  Therefore, no more than $3 \log T$
  agents of type $T-1$ can be selected. But if there are $6 T \log T
  \cdot T^{2(T-1)}$ draws, the optimal solution has value $3 T \log T
  \cdot T^{2(T-1)}$, and the algorithm has value no more than $3 \log
  T \cdot T^{2(T-1)} (1+o(1))$, which is less by roughly a factor of
  $T$ less.

  \medskip

  Thus, there is no $o(T)$ competitive algorithm, and the number of
  draws is at most $O(T^{2T-1} \log T)$. That is, if $n$ denotes the
  number of draws, there is no $o(\log n / \log \log n)$-competitve
  algorithm.  Using Yao's minimax principle, a similar argument can be
  extended to show that no randomized algorithm can obtain good
  approximations; we omit details from this extended abstract.
\end{proof}

\section{Display Ad Allocation and Fairness}\label{sec:fair} 

Other metrics besides efficiency play an important role in measuring
the quality of an allocation. In this section, we focus on the Display
Ad Allocation (DA) problem. Recall that in the DA problem, a set $J$
of $m$ advertisers have paid a website publisher in advance for their
ads to be shown to visitors to the website; for each advertiser $j \in
J$, their contract specifies an upper bound $n(j)$ on the number of
impressions they wish to pay for.  Each agent/impression has a set of
$m$ options corresponding to the $m$ advertisers, and must be assigned
to a single advertiser. If we assign impression $i$ to advertiser $j$,
it occupies one $j$'s $n(j)$ slots, and we obtain value $w_{ij}$.

In addition to the overall efficiency of the allocation, an important
consideration is its \emph{fairness} to the various advertisers; An
advertiser who does not get his ``fair share'' of impressions is
unlikely to purchase further contracts for impressions in the future.
Here, we propose a metric to capture the fairness of an allocation and
present algorithms to compute it.

Qualitatively, an allocation is ``fair'' if the advertisers are
treated fairly relative to each other.  As opposed to efficiency,
which is easily quantified as the sum of individual advertiser values,
fairness is more problematic, as it is inherently a relative (rather
than purely additive) measure.  One natural option is to consider
``max-min'' fairness, where the goal is to maximize the minimum
efficiency among the
advertisers~\cite{KRT01,KK06,LMMS04,BS06,AS07,BCG09,CCK09}.  While
useful in some contexts, in this application max-min fairness gives
too much attention to the most difficult-to-satisfy advertiser,
abandoning overall performance.  Given the diversity of demands,
impression targeting criteria and edge weights, a more flexible
fairness measure is needed.  In addition, the total weight of
impressions assigned to an advertiser depends not only on the eligible
set of impressions for that advertiser, but also the competition among
advertisers, i.e., if many advertisers are eligible for the same set
of (high-quality) impressions, none of these advertisers can get all
of these impressions, and these (high-quality) impressions should be
divided in some manner among the eligible advertisers.

Since this competition is intimately related to the structure of the
instance, it is difficult to quantify fairness in this context in a
universal way; thus, in order to define a fairness measure capturing
the above aspects, we first define an {\em ideal (offline) fair
  allocation} by taking into account advertisers competing for the
same set of impressions.  We define this allocation algorithmically,
i.e., it is a function of the problem instance.  We then compute the
fairness of an arbitrary assignment of impressions to advertisers by
computing the distance of this allocation to this ideal fair
allocation.

More precisely, we define the fairness measure as follows: Given an
allocation $x_{ij}$ of impressions $i$ to advertisers $j$, let $v_j(x)
= \sum_{i\in I} w_{ij}x_{ij}$ for each $j\in J$ denote the value
assigned to advertiser $j$.  The $v_j(x)$ can be defined for both 0/1
and fractional allocations $x$ in which $0\le x_{ij}\le 1$.  (In a
fractional allocation, the advertisers ``share'' the impression, which
one could interpret as a random allocation according to the implied
distribution.) For an allocation $x$, we roughly define the fairness
metric as the $l_1$ distance between $x$ and some ideal allocation
$x^*$, but where $x$ is normalized (scaled linearly) so that it has
the same efficiency as $x^*$. This scaling ensures that $x$ is judged
purely based on its relative efficiency among advertisers, rather than
on absolute efficiency.  We scale $x$ to match $x^*$ (rather than the
other way around) so that we may compare the fairness of different
allocations with a universal scale.  Formally, for an allocation $x$,
let $V(x) = \sum_{j \in J} v_j(x)$.  We define the {\em fairness
  measure} $f(x)$ as
\[
f(x) = \sum_{j\in J} \bigg \vert \frac{V(x^*)}{V(x)} v_j(x) - v_j(x^*)
\bigg \vert.
\]

Thus, the smaller $f(x)$ the fairer is allocation $x$.  Now, in order
to complete the definition of the fairness measure, it remains to
define the offline ideal fair allocation ${x^*}$.

\subsection{Offline Fair Allocations}
In this section, we discuss various natural offline fair allocations
${x^*}$ that can be used in the definition of fairness measure
defined above. As we discussed earlier, such ideal fair allocation
depends on the eligible set of impressions, and the set of advertisers
competing for the same impressions. Let $I(j)$ be the set of eligible
impressions for advertiser $j$ with demand $n(j)$. Assuming that
weights $w_{ij}$ capture the quality/relevance of impression $i$ for
advertiser $j$, in an ideal situation, advertiser $j$ would like to
get all the $n(j)$ impressions in $I(j)$ with the maximum weight. In
other words, ordering impressions in $I(j)$ in the non-increasing
order of their weight to $j$, advertiser $j$ would ideally want to get
a {\em prefix} of $n(j)$ impressions in this order. However, it might
not be possible for each advertiser $j$ to get a prefix of the first
$n(j)$ impressions in his ideal order, since an impression $i$ may
appear in the prefix of several advertisers. In such situations, we
should resolve the conflict (competition) of {\em interested
  advertisers} for this impression $i$ in a {\em fair} way, and extend the prefix of the affected advertisers.

Since we allow the offline fair allocation $x^*$ to be fractional,
this competition may be resolved by {\em sharing} each impression
among all interested advertisers.  A natural fair way of sharing an
impression $i$ among a set $J(i)$ of interested advertisers is to
divide this impression $i$ {\em equally} among all advertisers in
$J(i)$, i.e, each advertiser $a\in J(i)$ gets a fraction ${1\over
  \vert J(i) \vert}$ of impression $i$.  We call this method the {\em
  equal sharing} method (we discuss other sharing methods later.)

Given an arbitrary sharing policy like the {\em equal sharing policy}
defined above, we formally define the notion of a fair allocation
$x^*$ in terms of this policy:

\begin{definition}\label{def:fair}
  Let $H$ be a {\em sharing policy} mapping the advertisers $j$
  interested in impression $i$ to a fractional allocation
  $\{x_{ij}\}_{j \in J}$.  A fractional allocation ${x}^*$ is {\em
    fair under $H$}, if
\begin{itemize}
\item for each advertiser $j$, the set of impressions that $j$ is
  interested in is a prefix of all impressions (ordered by $w_{ij}$),
\item the allocation $x^*$ represents the policy $H$ applied to each
  impression, and
\item each advertiser is either interested in all impressions, or is
  receiving at least $n(j)$ impressions under ${x}^*$.
\end{itemize}
\end{definition}

An alternate way of thinking of a fair allocation is in terms of a
game, where each advertiser declares a set of impressions they are
interested in, and the mechanism then applies $H$ to these declarations.
A fair allocation is then any Nash equilibrium of this game.

We call a fair allocation under equal sharing an {\em equal share
  allocation}.  One can compute one such fair allocation $x^*$, in an
iterative method, as follows:

\bigskip
\noindent\underline{\em Fair Allocation algorithm}
\begin{enumerate}
\item
Maintain allocation variables $\{x_{ij} : i \in I, j \in J\}$ and
prefix ``pointers'' $\{ p(j) : j \in J\}$.  Initialize all $x_{ij} =
0$ and $p(j) = 0$.
\item Until all advertisers are satisfied, i.e., either $\sum_{i\in I}
  x_{ij} \ge n(j)$ or $p(j) = n$:
\begin{enumerate}
\item Let $j$ be some unsatisfied advertiser.  Increase $p(j)$ by one,
  and let $i$ be the $p(j)$-th best impression in $j$'s preference
  order.  Also, let $J(i)$ be the set of all advertisers $j'$ for whom
  $i$ is among the $p(j')$-th best impressions for that advertiser
  (and note $j \in J(i)$).  Set $x_{ij}$ according to $H$ for all $j
  \in J(i)$.  (For example, under equal sharing, we set $x_{ij} = {1
    \over \vert J(i) \vert}$.)
\end{enumerate}
\end{enumerate}

Note that there could be many different fair allocations, each with
different efficiency. For example suppose there were two impressions
$I = \{1,2\}$, and two advertisers $J = \{a,b\}$, each with capacity
one. Now suppose $w_{1,a} = 100$, $w_{2,a} = 10$, $w_{1,b} = 4$,
$w_{2,b} = 6$.  Then $x_{1,a} = x_{2,a} = x_{1,b} = x_{2,b} = \frac 1
2$ is a fair allocation with value $60$; the allocation $x_{1,a} =
x_{2,b} = 1$, $x_{1,b} = x_{2,a} = 0$ is also fair and has value
$106$.  However the following theorem shows that the given algorithm
always finds the most efficient fair allocation.

\begin{theorem} \label{thm:equalshare} The {\em Fair Allocation
    algorithm} runs in polynomial time and computes an offline fair
  allocation under any sharing policy where adding an advertiser to
  the set of interested advertisers does not increase the share of any
  other advertiser.  Moreover, for any sharing policy $H$, this
  algorithm produces the most efficient allocation among all fair
  allocations under $H$.
\end{theorem}

\begin{proof}
  In each iteration of the algorithm, one pointer advances, and
  therefore the number of iterations is bounded by the number of edges
  in the allocation graph, which is polynomial. To see that it
  produces the most efficient allocation under any sharing policy $H$,
  we use the following definition: Let $x_1$ and $x_2$ be two fair
  allocations under $H$, and let $I_1(j), I_2(j)$ be the set of
  impressions advertiser $j$ is interested in for $x_1$ and $x_2$
  respectively. Now, $x_1$ is said to be \emph{shorter than} $x_2$ if
  $I_1(j) \subseteq I_2(j)$ for each advertiser $j$, and the
  containment is strict for some advertiser. 

  We show that there exists a unique shortest fair allocation: Let
  $x_1$ and $x_2$ be fair allocations under $H$ such that neither is
  shorter than the other, and define a new allocation in which each
  advertiser $j$ is interested in impressions $I^*(j) = I_1(j) \cap
  I_2(j)$ (i.e., $j$ requests the shorter prefix from $I_1(j)$ and
  $I_2(j)$). It is easy to see that the number of impressions $j$
  receives in the new allocation is at least the minimum of the number
  it receives in $x_1$ and $x_2$, and hence at least
  $n(j)$\footnote{This may be less than $n(j)$ if $j$ is interested in
    all impressions in both $x_1$ and $x_2$, but in this case, $j$ is
    interested in all impressions in the new allocation.}.

  Let $x^*$ be the unique shortest allocation, and let $I^*(j)$ denote
  the set of impressions advertiser $j$ is interested in. To see that
  our algorithm returns $x^*$, consider the first step of the
  algorithm in which $p(j)$ moves beyond $I^*(j)$ for any advertiser
  $j$: Since each other advertiser has so far requested a set of
  impressions no larger than the set it requests for $x^*$ and $j$
  receives $n(j)$ impressions under $x^*$, $j$ already receives $n(j)$
  impressions under our algorithm. Thus, $j$ would not have been
  unsatisfied and the prefix pointer $p(j)$ would not have been
  incremented, a contradiction.

  Finally, it is easy to verify that for any fair allocations $x_1, x_2$,
  if $x_1$ is shorter than $x_2$, then $x_1$ is at least as efficient
  as $x_2$. This follows from the facts that $I_1(j)$ is a prefix of
  $I_2(j)$ when impressions are ordered by $w_{ij}$, and that for each
  impression in $I_1(j)$, $j$ receives a share in $x_1$ that is at
  least as large as it does in $x_2$.
\end{proof}

We can describe other variants of this fair allocation by altering how
we share an impression among those interested in it.  One natural way
to do this is to divide an impression $i$ among all advertisers in
$J(i)$, {\em proportional to the weight} of impression $i$ for these
advertisers, i.e, each advertiser $j \in J(i)$ gets a fraction $w_{ij}
\over \sum_{j'\in J(i)} w_{ij'}$ of impression $i$. We call this
sharing method, the {\em proportional sharing} method.  By a similar
argument to that of Theorem~\ref{thm:equalshare}, we can show that the
algorithm runs in polynomial time.  Later, we will discuss the
efficiency of such a fair allocation.

Inspired by the idea of stable matchings, one can also define an
extreme way of sharing an impression $i$ among advertisers by
introducing a strict preference order for each impression, and giving
this impression $i$ to an interested advertiser in $J(i)$ with the
highest priority in the preference order of impression $i$. In
particular, a natural preference order for impression $i$ is to order
advertisers in non-increasing order of their weight for impression
$i$, i.e, $w_{ij_1} \ge w_{ij_2} \ge \ldots, \ge w_{ij_k}$.  We call
this sharing method, the {\em stable-matching sharing}
method. Although this allocation may have some features that do not
seem ``fair'', an advantage of this definition is that it achieves
{\em approximate efficiency}.

\begin{theorem} The efficiency of the stable-matching sharing method
  is at least $1\over 2$ of the allocation with maximum
  efficiency. Moreover, the efficiency of the equal-sharing and the
  proportional-sharing method can be arbitrarily far from the optimum.
\end{theorem}
\begin{proof}
  First, we observe that the equal- and proportional-sharing methods
  can result in a fair allocation with arbitrarily bad performance.
  Consider $K^2$ advertisers; advertiser $i$ has value $\eps < {1\over
    K^2}$ for impression $i$.  In addition, there is one special
  impression; advertiser $1$ has value $K$ for it, and all other
  advertisers have value $1$ for it.  Every advertiser wants 1
  impression. The maximum weight matching gets value at least $K$, by
  giving the special impression to advertiser $1$, and giving every
  other advertiser $i$ impression $i$.  The proportional sharing
  method implies that for the special impression (everyone's first
  choice), the total value for people who want it is $K + (K^2-1)$.
  As a result, the first advertiser only gets roughly $1/K$ of the
  special impression, and therefore, the fair matching with
  proportional sharing is not efficient. The same example shows that
  the equal sharing method may also result in an inefficient fair
  allocation.

  Now, for the stable-matching sharing method, one can verify that the
  fair allocation in this setting is equivalent to a Nash equilibrium
  of a market sharing game defined as follows: The players are
  advertisers and markets are impressions $I$. Each player $j$ can
  play a subset $S_j\subset I$ of size at most $n(j)$ of impressions,
  and the weight of each impression goes to a player who has this
  impression in her item set $S_j$. It is not hard to show that this
  game is a valid-utility game with a submodular social function equal
  to the weight of the corresponding matching in an equilibrium. It
  follows by a known result of Vetta~\cite{V02}, that the price of
  anarchy of Nash equilibria in these games is $1\over 2$, and this
  implies that the value of the fair matching with stable-matching
  sharing rule is at least $1\over 2$ of the optimum solution.
\end{proof}

Even though, in the worst case, the equal sharing method may result in
an arbitrarily inefficient allocation, in practice it seems that the
efficiency of the equal-sharing allocation is on the same order of
magnitude as the optimum efficiency (we will show this in our
experiments in Section~\ref{sec:experimental}).

\section{Online Heuristic Algorithms}\label{sec:heuristics} In this
section, we list a set of online competitive algorithms for the
display ad allocation problem that we will study in our experimental
evaluation. Some of these algorithms are already known and analyzed
for their theoretical worst-case performance~\cite{WINE09}, and some
are combinations of the algorithms studied in this paper.

All of these algorithms can be described based on the primal and dual
linear programming formulations for the display ad allocation problem
studied in Section~\ref{sec:ptas}.  In fact, we can interpret these
algorithms as simultaneous constructing feasible solutions to the
primal and dual LPs, using the following outline:
\begin{itemize}
\item For each advertiser $j$, initialize dual variable $\beta_j$
to $0$.

\item When an impression $i$ arrives online, assign $i$ to the
  advertiser $j' \in J$ that maximizes $w_{ij} - \beta_j$. (If this
  value is negative for each $j$, we may leave impression $i$
  unassigned.) Set $x_{ij'} = 1$.

\item If $j'$ previously had $n(j')$ impressions assigned, let $i'$ be
  the least valuable of these; set $x_{i'j'} = 0$.

\item In the dual solution, set $z_i = w_{ij'} - \beta_{j'}$ and
  increase $\beta_{j'}$ using an appropriate \emph{update rule} (see
  below); different update rules give rise to different
  algorithms/allocations.
\end{itemize}

In order to define different variants of this algorithm, we should
define the update rule for the dual variables.

\begin{enumerate}
\item {\bf Greedy Algorithm {\greedy}:} For each advertiser $j$,
  $\beta_j$ is the weight of the lightest impression among the $n(j)$
  heaviest impressions currently assigned to $j$. That is, $\beta_j$
  is the weight of the impression which will be discarded if $j$
  receives a new high-value impression. An equivalent interpretation
  of this algorithm is to assign each impression to the advertiser
  with the maximum marginal increase in the weight of the matching.

\item {\bf Uniform Average (\pdavg): } For each advertiser $j$,
  $\beta_j$ is the average weight of the $n(j)$ most valuable
  impressions currently assigned to $j$. If $j$ has fewer than $n(j)$
  assigned impressions, $\beta_j$ is the ratio between the total
  weight of assigned impressions and $n(j)$.

\item {\bf Exponential Weighted Average (\pdexp) :} For each
  advertiser $j$, $\beta_j$ is an ``exponentially weighted average''
  of the $n(j)$ most valuable impressions, defined as follows: Let
  $w_1, w_2, \ldots w_{n(j)}$ be the weights of impressions currently
  assigned to advertiser $j$, sorted in non-increasing order.\\ Let
  $\beta_j = \frac{1}{n(j) \cdot \left((1+1/n(j))^{n(j)} - 1 \right)}
  \sum_{k=1}^{n(j)} w_k \left(1 + \frac{1}{n(j)} \right)^{k-1}$.
\end{enumerate}

In the previous paper~\cite{WINE09}, the authors prove that
{\greedy}, {\pdavg}, and {\pdexp} algorithms achieve worst-case
competitive ratios of $1\over 2$, $1\over 2$, and $1-{1\over e}$
respectively. In this paper, we will compare these online algorithms
with a training-based algorithm which is based on computing dual
variables $\beta$ based on some sample data, and then applying these
fixed dual variables for the rest of the algorithm.

We also study a hybrid algorithm, called {\hybrid}, combining the
training-based online algorithm from Section~\ref{sec:ptas} and a pure
online algorithm. This algorithm is inspired by ideas of Mahdian,
Nazerzadeh, and Saberi~\cite{MNS}.  In this hybrid algorithm, we set
$\beta_j$ for each advertiser $j$ to be a convex combination of two
algorithms: Let $\beta^1_j$ be the dual variable learnt by the
training-based algorithm and remaining fixed throughout the algorithm
and let $\beta^2_j$ be the dual variable as currently used by
\pdavg. We set $\beta_j = \alpha \beta^1_j + (1-\alpha) \beta^2_j$ for
some $0 \leq \alpha \leq 1.$ Initially we set $\alpha = 1$ and we
decrease $\alpha$ gradually throughout the algorithm until it hits
0. Thus the algorithm starts using the fixed $\beta^1$ values and
gradually switches to the $\beta^2$ values, which in turn change as
impressions are processed.  As we will see in the experimental
results, this algorithm outperforms both the training-based and the
\pdavg algorithm.

\section{Experimental Evaluation}\label{sec:experimental}
In this section, we discuss the experimental results comparing the
efficiency and fairness of the algorithms discussed in this paper.

{\bf \noindent Data Set.}  Our sample data set consists of
(a uniform sample) of a set of arriving impressions and a set of advertisers
for six different publishers (A-F) over one week.  The
number of arriving impression varies from 200,000 to 1,500,000
impressions, and the number of advertisers per publisher varied from
100 to 2,600 advertisers (see Table~\ref{tab:info}).
Each impression is tagged with their set of
eligible advertisers, and an {\em edge weight} for each eligible
advertiser capturing the ``quality score'' for assigning this
impression to this advertiser.  The distribution
of {\em edge weights}  approximately follows the
log-normal distribution.

\begin{table}
\begin{center}
\begin{tabular}{|ccccccc|}
\hline
\tiny{Publishers} & A & B  &  C  & D  & E & F\\
\hline
$m$ & 109 &  1117 & 636 & 1586 & 2585 & 1113 \\
$n$ & \tiny{$5\times 10^5$} & \tiny{$4\times 10^5$} &\tiny{$2\times 10^5$} &\tiny{$9\times 10^5$} & \tiny{$1.5\times 10^6$} &\tiny{$4\times 10^5$} \\
\hline
\end{tabular}
\caption{Number of advertisers and number of arriving 
impressions for each of the six publishers.}
\label{tab:info}
\end{center}
\end{table}

\begin{table}
\begin{center}
\begin{tabular}{|cccccccc|}
\hline
Publishers & A & B  &  C  & D  & E & F &  Avg\\
\hline
\lpweight & 100 & 100 & 100 & 100 & 100 & 100 & 100 \\
\fair & 88.2 & 98.4 & 73.6 & 42.3 & 74.6 & 53.3 & 71.7 \\
\hline
\dualbase & 85 & 93 & 85.7 & 74 & 91.8 & 93.5 & 87.2 \\
\hybrid & 85 & 93.8 & 95.2 & 73.8 & 92.7 & 93.5 & 89 \\
\hline
\greedy & 64 & 90.5 & 69.7 & 53.6 & 55 & 86.2 & 69.8 \\
\pdavg & 72 & 93.2 & 75.3 & 65.3 & 71.7 & 89.5 & 77.8 \\
\pdexp & 72.6 & 89.7 & 73.9 & 90.8 & 72.6 & 96.3 & 82.6 \\ 
\hline
\end{tabular}
\caption{Normalized efficiency of different algorithms for different publishers and
averaged over all publishers. 
All numbers are normalized such that the efficiency of $\mbox{OPT}=$\lpweight is 
100.}
\label{tab:efficiency}
\end{center}
\end{table}

\begin{table}
\begin{center}
\begin{tabular}{|cccccccc|}
\hline
Publishers & A &    B   &  C  &    D  & E & F &    Avg\\
\hline
\fair & 0 & 0 & 0 & 0 & 0 & 0 & 0 \\
\lpweight & 34.6 & 47.7 & 98.8 & 100 & 70.3 & 90.1 & 73.6 \\
\hline
\dualbase & 69.5 & 62.5 & 96.7 & 43.1 & 87.9 & 88.6 & 74.7 \\
\hybrid & 69.4 & 63.1 & 100 & 41.9 & 83.7 & 88.6 & 74.5 \\
\hline
\greedy & 100 & 100 & 98.6 & 45 & 100 & 100 & 90.6 \\
\pdavg & 73 & 72 & 82.7 & 31.7 & 91.9 & 85.3 & 72.8 \\
\pdexp & 69.7 & 59.5 & 86.1 & 71 & 88.8 & 100 & 79.2 \\
\hline
\end{tabular}
\caption{Normalized fairness of different algorithms for different publishers and averaged over all publishers. 
All numbers of each column are normalized between zero and 100, where
$0$ is the most fair solution.}
\label{tab:fairness}
\end{center}
\end{table}

\smallskip {\bf \noindent The Algorithms.}  We examine (a) three pure
online algorithms, (b) two training-based online algorithms, and (c)
two offline algorithms.  (a) The pure online algorithms are
{\greedy}, {\pdavg}, and {\pdexp}; see
Section~\ref{sec:heuristics}.  (b) For the training-based online
algorithm we use the primal-dual based algorithm from
Section~\ref{sec:ptas}, called {\dualbase}, and the \hybrid
algorithm from Section~\ref{sec:heuristics}. For both of them we
construct the training data as follows: For each data set, sample 1\%
of the impressions uniformly and use it for training.  The remaining
99\% of the impressions are used as a test set.  With this sampling
step we hope to proxy the random order model, since in the random
order model a sample of the whole data is equivalent to a sample from
the beginning part of the sequence.
(c) As offline algorithms we use the fair algorithm using {\em equal
  sharing}, called \fair and described in Section~\ref{sec:fair},
and the algorithm \lpweight, which computes the
optimal efficient assignment (i.e. the maximum weight b-matching). The
latter is computed by solving the primal LP using the GLPK LP solver.


\begin{figure}[t]
\includegraphics[width=7cm,height=2.25in]{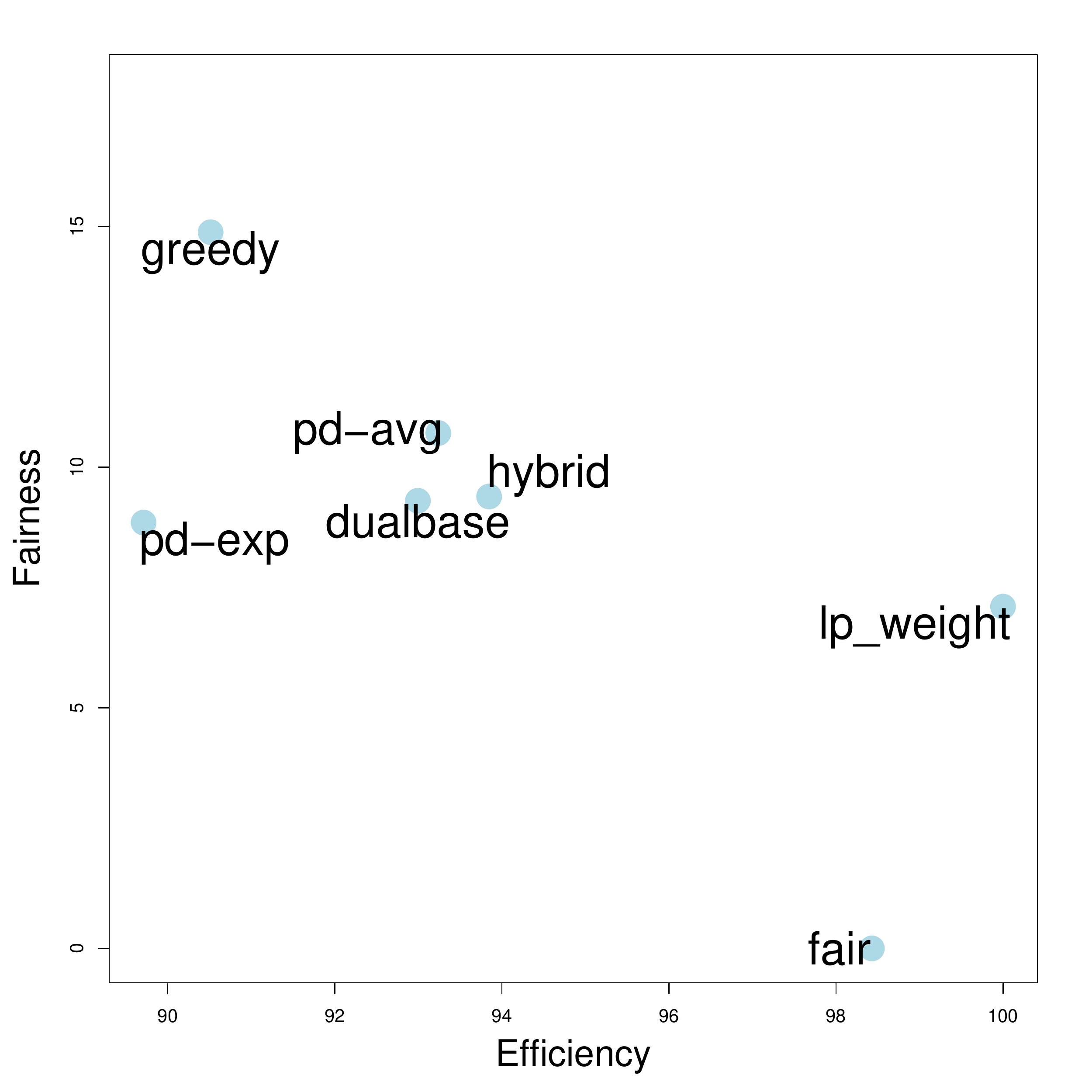}
\includegraphics[width=7cm,height=2.25in]{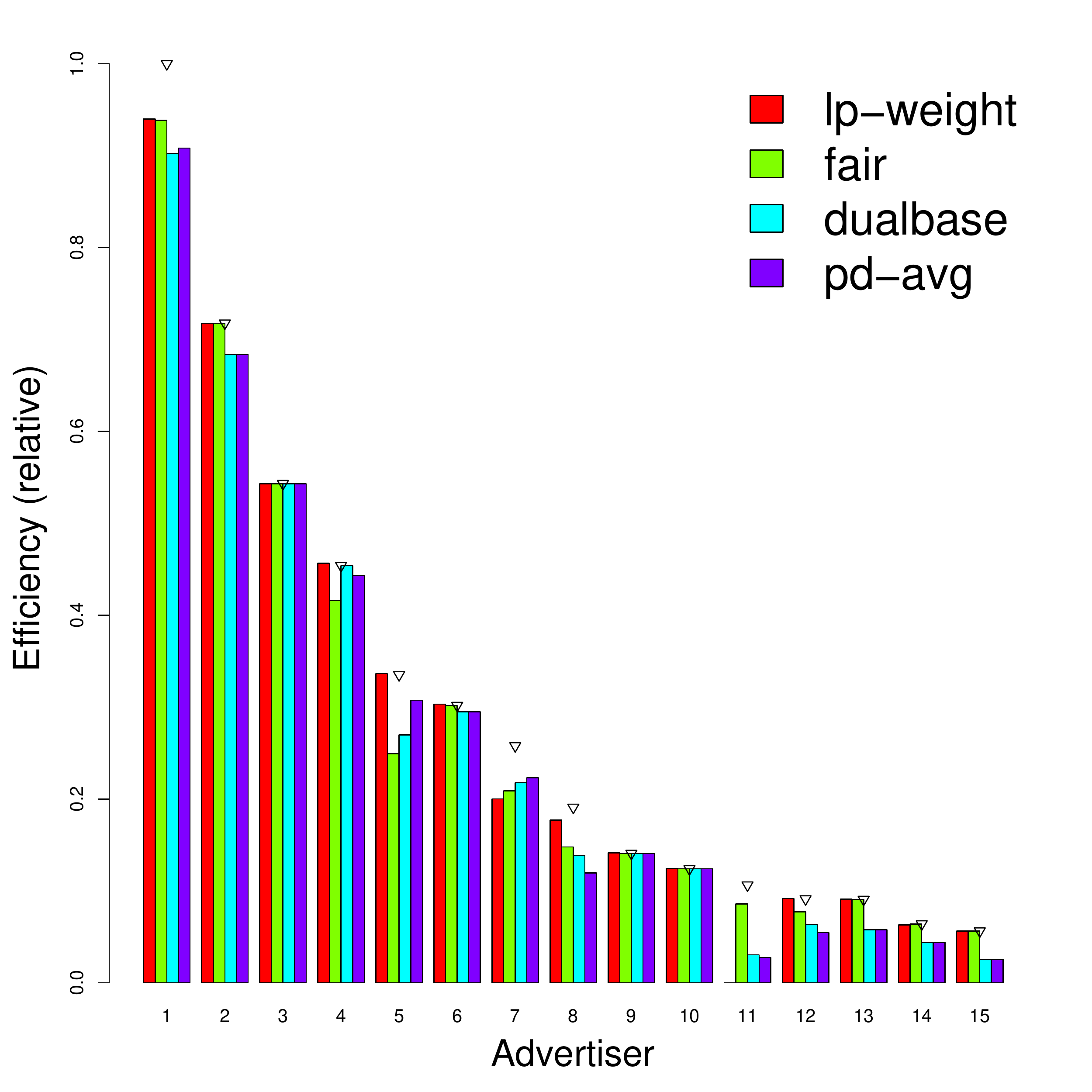}
\caption{
Efficiency and fairness of algorithms for Publisher B (left). 
Comparison of efficiency of different advertisers for Publisher B (right).
Advertisers are sorted by their maximum possible efficiency 
(given by the inverted triangle).
}
\label{fig:ef1}
\vspace{-2mm}
\end{figure}

\begin{figure}[tbh]
\includegraphics[width=7cm,height=2.25in]{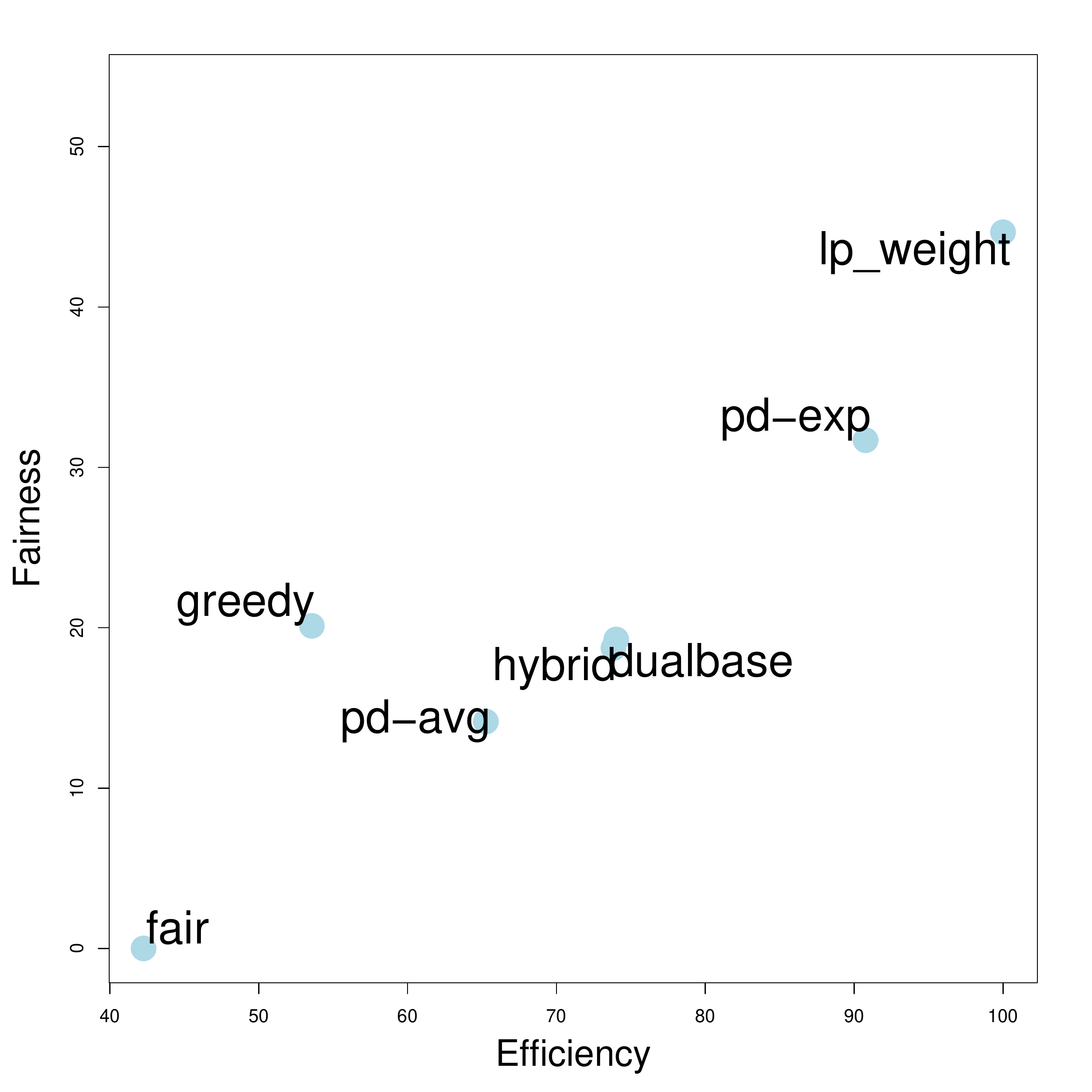}
\includegraphics[width=7cm,height=2.25in]{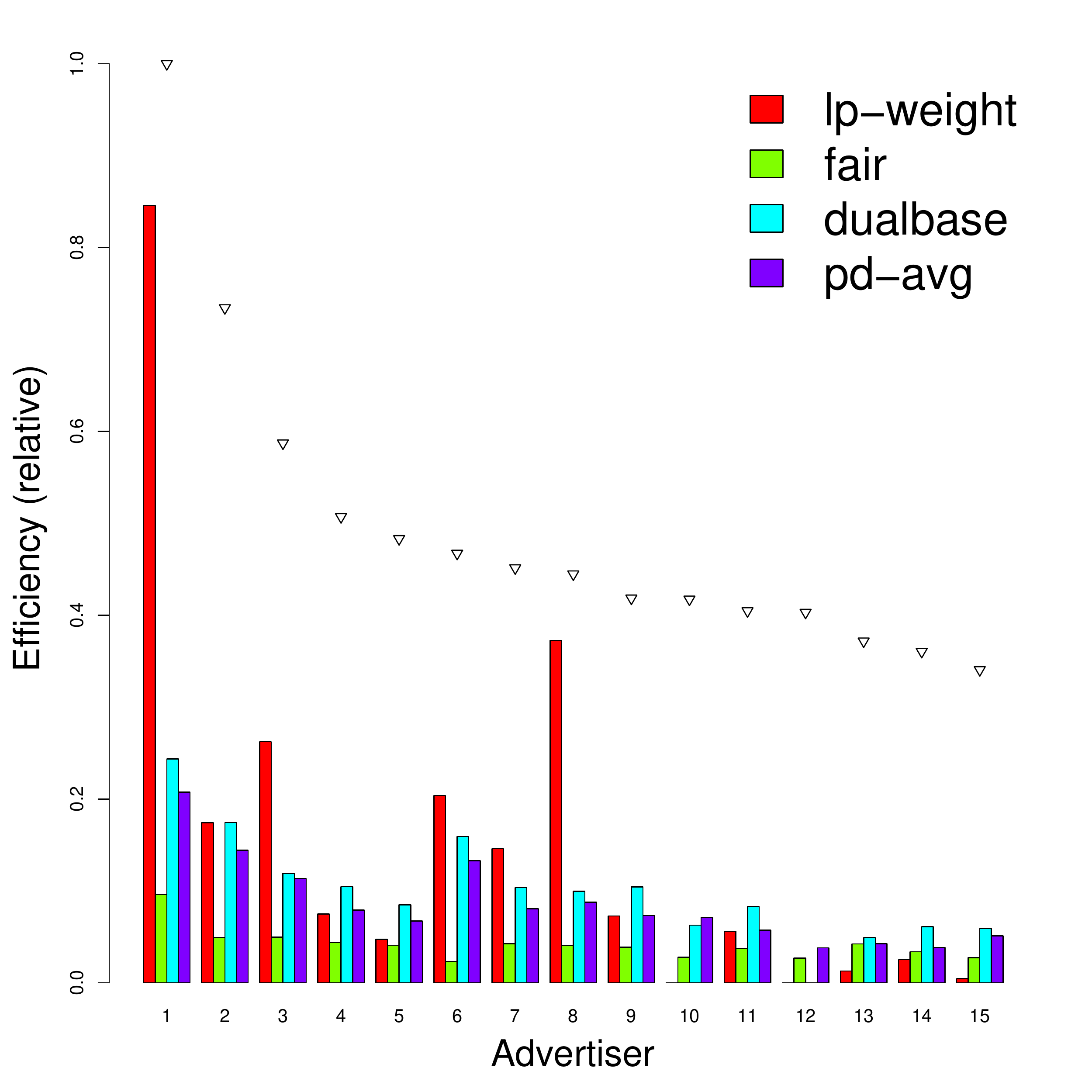}
\caption{
Efficiency and fairness of algorithms for Publisher D (left). 
Comparison of efficiency of different advertisers for Publisher D (right).
Advertisers are sorted by their maximum possible efficiency 
(given by the inverted triangle).}
\label{fig:comp}
\vspace{-2mm}
\end{figure}

\begin{figure}[tbh]
\includegraphics[width=7cm,height=2.25in]{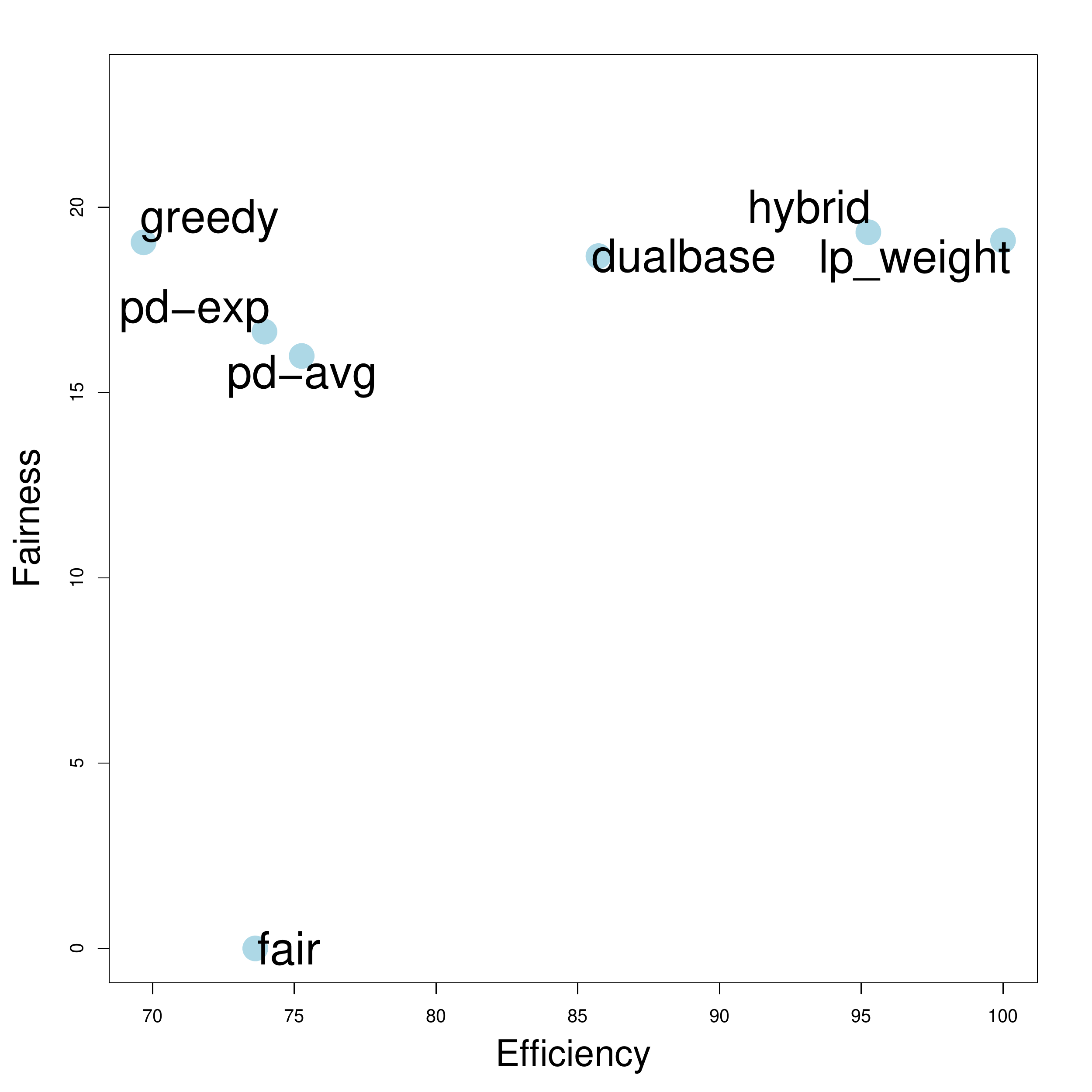}
\includegraphics[width=7cm,height=2.25in]{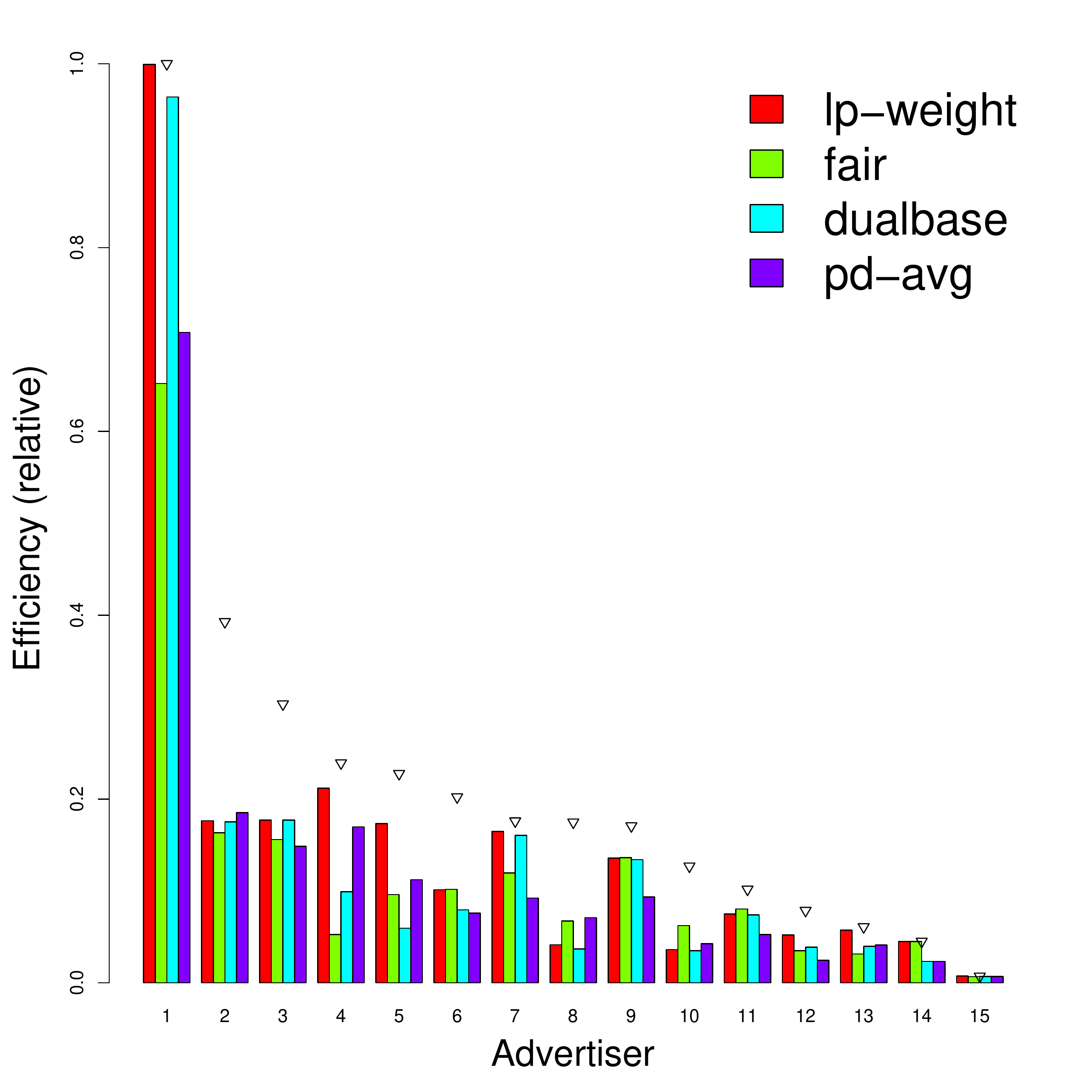}
\caption{
Efficiency and fairness of algorithms for Publisher C (left). 
Comparison of efficiency of different advertisers for Publisher C (right).
Advertisers are sorted by their maximum possible efficiency 
(given by the inverted triangle).}
\label{fig:ef2}
\vspace{-4mm}
\end{figure}

\smallskip {\bf \noindent Experimental Results.}  The efficiency and
(normalized) fairness of the output of each of the algorithms are
summarized in Tables~\ref{tab:efficiency} and \ref{tab:fairness}. The
results for three representative publishers are additionally depicted
in Figures ~\ref{fig:ef1}, ~\ref{fig:ef2}, and ~\ref{fig:comp}.
Recall that we normalized efficiency so that the efficiency-optimal
algorithm \lpweight has efficiency $100$. Table~\ref{tab:efficiency}
shows that (1) the training-based algorithms clearly outperform the
pure online algorithms, (2) of the pure online algorithms, both
\pdavg and \pdexp outperform \greedy, and (3) \hybrid and
\dualbase perform very similarly, except for one publisher where
\hybrid clearly outperforms \dualbase.

Table~\ref{tab:fairness} shows {\em normalized} fairness. Since the
value of fairness depends on the values assigned to advertisers and
different publishers have different advertisers, we normalized the
fairness values for each publisher so that the {\em least} fair
algorithm achieves a score of 100 and algorithm \fair achieves a
score of 0. Normalizing allows us to compute the average over
different publishers.  The results in the table indicate that
\greedy is the least fair algorithm. The remaining algorithms,
including \lpweight, perform roughly the same, though their
performance differs over different publishers.

Figures~\ref{fig:ef1}--\ref{fig:ef2} plot efficiency
vs. (unnormalized) fairness and they show additionally the efficiency
achieved for the top 10 advertisers for four of the algorithms. The
inverted triangle above each advertiser represents the maximum
possible efficiency for this advertiser if the other advertisers did
{\em not} exist. There are three rough categories and the publishers
for which we show this data each represent a different category: For
publisher B in Figure~\ref{fig:ef1} the maximum possible efficiency of
the top advertisers is almost the same as the efficiency achieved by
all algorithms.  This publisher is undersold with little competition
between the advertisers.  Thus, for this publisher, the choice of
algorithm does not heavily influence efficiency.
Table~\ref{tab:efficiency}, shows that for publisher B all algorithms,
including \fair, achieve an efficiency of 90 or above. The situation
is similar for publisher A (not shown). In both settings \fair has an
impressively high efficiency and \lpweight achieves a good fairness
value. In such a low-competitive situation the online algorithms are
in a clear disadvantage over the offline algorithms. Also the
training-based online algorithms outperform the pure online algorithms
as they can leverage their knowledge about the data to construct a
more efficient and more fair solution.

Publisher D in Figure~\ref{fig:comp} shows the other extreme: Here the
maximum possible efficiency of the top advertisers is much larger than
the efficiency achieved by any of the algorithms, including the
optimum \lpweight. This publisher has a lot of competition between
the advertisers.  Publisher F (not shown) is in a similar, but a bit
less extreme situation.  In both cases, the choice of an algorithm has
a large influence on the efficiency, as can be seen in
Table~\ref{tab:efficiency}: Algorithm \fair distributes the weight
more evenly across the advertisers than any of the other algorithms,
but also achieves only an efficiency of about 42, resp.~53. Algorithm
\lpweight, on the other side, generates a very uneven distribution
of weights, giving a lot of efficiency to advertiser 1 and 8.  For
both publishers \pdexp clearly outperforms the non-optimal
algorithms. \pdexp also has better theoretical performance.

Finally publishers C (in Figure~\ref{fig:ef2}) and E (not shown)
represent the ``in-between'' situation: The maximum possible
efficiency of the top advertisers is somewhat larger than the
efficiency achieved by the algorithms, but there is not a large
gap. In both cases the training-based algorithms clearly outperform
the pure online algorithms in efficiency. Thus, this is the situation
where learning clearly helps in terms of efficiency.

\bigskip
\noindent Overall we draw the following conclusions:

Algorithm \pdavg generally achieves much better efficiency and
fairness than \greedy, even though both algorithms are $1\over
2$-competitive in the worst case. Algorithm \pdavg also results in
the best fair solution among all algorithms and \greedy has the
worst fairness measure.

The training-based algorithms generally achieve higher efficiency than
the pure online algorithms, especially in settings that are not too
extreme, i.e., oversold or undersold. On average, \dualbase improves
12\% over \pdavg, and 5\% over \pdexp.  Furthermore, \hybrid has
a marginal improvement (of 2\% on average, and upto 10\%) over
\dualbase, mostly based on a big improvement for one publisher.

Though the worst-case competitive analysis of \pdexp is much better
than \pdavg, this algorithm showed only $5\%$ overall improvement
over \pdavg, and in one case showed a significant loss in
efficiency. However, in highly competitive settings, \pdexp gives
large improvements.
 
\section{Concluding Remarks}
In this paper, we give a training-based algorithm for online
allocation, and prove that in the random-order stochastic model, it
achieves a $(1-\eps)$ approximation to the optimal solution under
mild assumptions. 

We also considered the Display Ad Allocation problem from both a
theoretical and empirical perspective, studying fairness in addition
to efficiency.  We introduced different notions of offline fair
allocations, and present a new fairness measure as a distance to such
offline fair allocations. Finally, we performed an experimental
evaluation of our training-based algorithm, along with previously
studied online algorithms and some hybrid algorithms. We compared
their performance on data sets from real display ad allocation
problems; our experiments show that among the pure online algorithms
designed for worst-case inputs, {\pdavg} performs reasonably well in
terms of both efficiency and fairness, and \pdexp gives large
improvements for more difficult instances. The training-based
algorithm outperforms {\pdavg} and \pdexp by a large factor, and
combining pure online and training-based methods in a hybrid algorithm
improves the efficiency further.

This paper motivates many open problems to explore: (i) Can we achieve
an algorithm that is simultaneously good both in the worst case and in
stochastic settings? Such an algorithm would be of use when the actual
distribution of agents is different from the one predicted/learnt from
a sample; in the display ad setting, this occurs when there is a
sudden spike in traffic to a website, perhaps in response to a
breaking news event, or links from an extremely high-traffic
source. (ii) Can we design an online allocation algorithm that
provably achieves approximate efficiency and approximate fairness (for
some an appropriate notion of fairness) at the same time? (iii) Can we
prove that in certain settings that appear in practice, the \pdavg
algorithm achieves an improved approximation factor (i.e., better than
$1\over 2$)?  (iv) Can we extend the online stochastic algorithm
studied in this paper to other stochastic process models such as
Markov-based stochastic models?  Answering these questions is an
interesting subject of future research.

\medskip 

{\bf \noindent Acknowledgments.} This paper is a followup of our
previous work with S. Muthukrishnan and Martin P\'{a}l, and some of the
results and discussions in this paper are inspired by our initial
discussions with them. We thank Martin and Muthu for their
contributing insights toward this paper. We also thank the Google
display ad team, and especially Scott Benson for helping us with data
sets used in this paper.

{\small{\bibliographystyle{abbrv} \bibliography{onlinematching}}}

\begin{thebibliography}{10}

\bibitem{AWY09}
S.~Agrawal, Z.~Wang, and Y.~Ye.
\newblock A dynamic near-optimal algorithm for online linear programming.
\newblock Working paper posted at http://www.stanford.edu/~yyye/.

\bibitem{AM09}
S.~Alaei and A.~Malekian.
\newblock Maximizing sequence-submodular functions, manuscript.
\newblock 2009.

\bibitem{AS07}
A.~Asadpour and A.~Saberi.
\newblock An approximation algorithm for max-min fair allocation of indivisible
  goods.
\newblock In {\em STOC}, pages 114--121, 2007.

\bibitem{AAP-Routing93}
B.~Awerbuch, Y.~Azar, and S.~Plotkin.
\newblock {Throughput-competitive on-line routing}.
\newblock In {\em FOCS}, volume~34, pages 32--40, 1993.

\bibitem{azar2008iaa}
Y.~Azar, B.~Birnbaum, A.~Karlin, C.~Mathieu, and C.~Nguyen.
\newblock {Improved Approximation Algorithms for Budgeted Allocations}.
\newblock In {\em ICALP}, 2008.

\bibitem{BIKK08}
M.~Babaioff, N.~Immorlica, D.~Kempe, and R.~Kleinberg.
\newblock Online auctions and generalized secretary problems.
\newblock {\em SIGecom Exchanges}, 7(2), 2008.

\bibitem{BS06}
N.~Bansal and M.~Sviridenko.
\newblock The santa claus problem.
\newblock In {\em STOC}, pages 31--40, 2006.

\bibitem{BCG09}
M.~Bateni, M.~Charikar, and V.~Guruswami.
\newblock Maxmin allocation via degree lower-bounded arborescences.
\newblock In {\em STOC}, 2009.

\bibitem{broder-tutorial}
A.~Broder.
\newblock Introduction to computational advertising, tutorial at wine '09.

\bibitem{broder-broadmatch}
A.~Z. Broder, P.~Ciccolo, E.~Gabrilovich, V.~Josifovski, D.~Metzler, L.~Riedel,
  and J.~Yuan.
\newblock Online expansion of rare queries for sponsored search.
\newblock In {\em WWW}, pages 511--520, 2009.

\bibitem{broder-IR-1}
A.~Z. Broder, M.~Fontoura, V.~Josifovski, and L.~Riedel.
\newblock A semantic approach to contextual advertising.
\newblock In {\em SIGIR}, pages 559--566, 2007.

\bibitem{buchbinder-jain-naor}
N.~Buchbinder, K.~Jain, and J.~Naor.
\newblock {Online Primal-Dual Algorithms for Maximizing Ad-Auctions Revenue}.
\newblock In {\em Proc. ESA}, page 253. Springer, 2007.

\bibitem{BN06}
N.~Buchbinder and J.~Naor.
\newblock Improved bounds for online routing and packing via a primal-dual
  approach.
\newblock In {\em FOCS}, pages 293--304, 2006.

\bibitem{CCK09}
D.~Chakrabarty, J.~Chuzhoy, and S.~Khanna.
\newblock On allocating goods to maximize fairness.
\newblock In {\em FOCS}, 2009.

\bibitem{chakrabarty2008aba}
D.~Chakrabarty and G.~Goel.
\newblock {On the approximability of budgeted allocations and improved lower
  bounds for submodular welfare maximization and GAP}.
\newblock In {\em Proc. FOCS}, pages 687--696, 2008.

\bibitem{CHMS10}
S.~Chawla, J.~D. Hartline, D.~Malec, and B.~Sivan.
\newblock Sequential posted pricing and multi-parameter mechanism design.
\newblock {\em CoRR}, To Appear, STOC 2010, 2010.

\bibitem{flavio2}
F.~Chierichetti, R.~Kumar, and S.~Vassilvitskii.
\newblock Similarity caching.
\newblock In {\em PODS}, pages 127--136, 2009.

\bibitem{devanur-hayes}
N.~Devanur and T.~Hayes.
\newblock The adwords problem: Online keyword matching with budgeted bidders
  under random permutations.
\newblock In {\em ACM EC}, 2009.

\bibitem{WINE09}
J.~Feldman, N.~Korula, V.~Mirrokni, S.~Muthukrishnan, and M.~Pal.
\newblock Online ad assignment with free disposal.
\newblock In {\em WINE}, 2009.

\bibitem{FMMM09}
J.~Feldman, A.~Mehta, V.~Mirrokni, and S.~Muthukrishnan.
\newblock Online stochastic matching: Beating 1 - 1/e.
\newblock In {\em FOCS}, 2009.

\bibitem{GMPV09}
A.~Ghosh, P.~McAfee, K.~Papineni, and S.~Vassilvitskii.
\newblock Bidding for representative allocations for display advertising.
\newblock In {\em WINE}, pages 208--219, 2009.

\bibitem{GRVZ09}
A.~Ghosh, B.~I.~P. Rubinstein, S.~Vassilvitskii, and M.~Zinkevich.
\newblock Adaptive bidding for display advertising.
\newblock In {\em WWW}, pages 251--260, 2009.

\bibitem{GM07}
G.~Goel and A.~Mehta.
\newblock Adwords auctions with decreasing valuation bids.
\newblock In {\em WINE}, pages 335--340, 2007.

\bibitem{GM08}
G.~Goel and A.~Mehta.
\newblock Online budgeted matching in random input models with applications to
  adwords.
\newblock In {\em SODA}, pages 982--991, 2008.

\bibitem{KVV}
R.~Karp, U.~Vazirani, and V.~Vazirani.
\newblock An optimal algorithm for online bipartite matching.
\newblock In {\em Proc. STOC}, 1990.

\bibitem{KRT01}
J.~M. Kleinberg, Y.~Rabani, and {\'E}.~Tardos.
\newblock Fairness in routing and load balancing.
\newblock {\em J. Comput. Syst. Sci.}, 63(1):2--20, 2001.

\bibitem{Kleinberg05}
R.~Kleinberg.
\newblock {A multiple-choice secretary algorithm with applications to online
  auctions}.
\newblock In {\em Proceedings of the sixteenth annual ACM-SIAM symposium on
  Discrete algorithms}, pages 630--631. Society for Industrial and Applied
  Mathematics, 2005.

\bibitem{KP09}
N.~Korula and M.~Pal.
\newblock Algorithms for secretary problems on graphs and hypergraphs.
\newblock In {\em ICALP}, 2009.

\bibitem{KK06}
A.~Kumar and J.~M. Kleinberg.
\newblock Fairness measures for resource allocation.
\newblock {\em SIAM J. Comput.}, 36(3):657--680, 2006.

\bibitem{LMMS04}
R.~Lipton, E.~Markakis, E.~Mossel, and A.~Saberi.
\newblock On approximately fair allocations of indivisible goods.
\newblock In {\em ACM EC}, 2004.

\bibitem{MNS}
M.~Mahdian, H.~Nazerzadeh, and A.~Saberi.
\newblock Allocating online advertisement space with unreliable estimates.
\newblock In {\em ACM EC}, pages 288--294, 2007.

\bibitem{msvv}
A.~Mehta, A.~Saberi, U.~Vazirani, and V.~Vazirani.
\newblock Adwords and generalized online matching.
\newblock In {\em FOCS}, 2005.

\bibitem{flavio1}
S.~Pandey, A.~Z. Broder, F.~Chierichetti, V.~Josifovski, R.~Kumar, and
  S.~Vassilvitskii.
\newblock Nearest-neighbor caching for content-match applications.
\newblock In {\em WWW}, pages 441--450, 2009.

\bibitem{srinivasan2008baf}
A.~Srinivasan.
\newblock {Budgeted Allocations in the Full-Information Setting}.
\newblock In {\em APPROX}, 2008.

\bibitem{V02}
A.~Vetta.
\newblock Nash equilibria in competitive societies, with applications to
  facility location, traffic routing and auctions.
\newblock In {\em FOCS}, 2002.

\end{thebibliography}


\end{document}